\documentclass[a4paper, 11pt]{article}
\usepackage[OT1]{fontenc}
\usepackage{amsthm,amsmath}
\usepackage[colorlinks,citecolor=blue,urlcolor=blue]{hyperref}

\usepackage{amssymb, graphicx, latexsym}
\usepackage{arydshln} %dashline

\usepackage[nottoc]{tocbibind}

\numberwithin{equation}{section}
\theoremstyle{plain}
\newtheorem{theorem}{Theorem}[section]
\newtheorem{proposition}{Proposition}[section]
\newtheorem{lemma}{Lemma}[section]
\newtheorem{remark}{Remark}[section]
\newtheorem{corollary}{Corollary}[section]

\usepackage{amsfonts}
\usepackage{authblk}

\usepackage[papersize={210mm,297mm},top=30mm, bottom=30mm, left=30mm , right=30mm]{geometry}

\usepackage{algorithm}

\usepackage{lipsum}
\usepackage{graphicx}
\usepackage{epstopdf}
\usepackage{algorithmic}

\ifpdf
  \DeclareGraphicsExtensions{.eps,.pdf,.png,.jpg}
\else
  \DeclareGraphicsExtensions{.eps}
\fi

\newcommand{\mb}{\mathbf}
\begin{document}

%%%%%%%%%%%%%%%%%%% TITLE PAGE %%%%%%%%%%%%%%%%%%%

\title{Is a Finite Intersection of Balls Covered by\\ a Finite Union of Balls in Euclidean Spaces ?}

\author[1]{
Vincent Runge\footnote{E-mail: runge.vincent@gmail.com}}

% List of institutions
\affil[1]{LaMME - Laboratoire de Math\'ematiques et Mod\'elisation d'Evry.\newline UEVE - Universit\'e d'Evry-Val-d'Essonne.}
\date{}
\maketitle

% Abstract of the paper
\begin{abstract}
Considering a finite intersection of balls and a finite union of other balls in an Euclidean space, we propose an exact method to test whether the intersection is covered by the union. We reformulate this problem into quadratic programming problems. For each problem, we study the intersection between a sphere and a Voronoi-like polyhedron. That way we get information about a possible overlap between the frontier of the union and the intersection of balls. If the polyhedra are non-degenerate, the initial nonconvex geometric problem, which is NP-hard in general, is tractable in polynomial time by convex optimization tools and vertex enumeration. Under some mild conditions the vertex enumeration can be skipped. Simulations highlight the accuracy and efficiency of our approach compared with competing algorithms in Python for nonconvex quadratically constrained quadratic programming. This work is motivated by an application in statistics to the problem of multidimensional changepoint detection using pruned dynamic programming algorithms.
\end{abstract}

%\begin{keywords}
Nonconvex quadratically constrained quadratic programming, ball covering problem, computational geometry, Voronoi-like polyhedron, vertex enumeration, polynomial time complexity.
%\end{keywords}

MS classification : 90C26, 52C17, 68U05, 62L10.

\section{Introduction}
\label{Section1}

\subsection{Problem description}
\label{ProbDesc}
We consider two finite sets of balls in an Euclidean space, $\Lambda = \{B_1,...,B_p\}$ and $\rm V = \{\mb{B}_1,...,\mb{B}_q\}$, with $p,q \in \mathbb{N}^*$ and arbitrary centers and radii. We introduce the intersection set $\mathcal{I} = \cap_{i=1}^{p}B_i$ and union set $\mathcal{U} = \cup_{j=1}^{q}\mb{B}_j$. Our problem consists in finding an exact and efficient method to decide whether the inclusion $\mathcal{I} \subset \mathcal{U}$ is true or false. 
Denoting by $\mathcal{U}^{\text{c}}$ the complement of $\mathcal{U}$, this problem is equivalent to the study of the emptiness of $\mathcal{I}\cap \mathcal{U}^{\text{c}}$ which is a challenging question, both theoretically and computationally due to the non-convexity of the $\mb{B}_j^{\text{c}}$ sets ($j=1,...,q$).\\

With $\mathbb{R}^n$ the $n$-dimensional Euclidean space, $n \ge 2$, the open balls $B_i$ and closed balls $\mb{B}_j$ are defined by their centers $c_i,\mb{c}_j \in \mathbb{R}^n$ and radii $R_i,\mb{R}_j \in \mathbb{R}^*_+$ respectively. Thus
$$B_i = \{x \in \mathbb{R}^n \,,\, \|x-c_{i}\|^2 < R_i^2\} \quad\hbox{and}\quad \mb{B}_j = \{x \in \mathbb{R}^n \,,\, \|x-\mb{c}_{i}\|^2 \le \mb{R}_j^2\}\,,$$
where $\|x-c_{i}\|^2 = \sum_{k=1}^{n} (x_k-c_{ik})^2$, with $x = (x_1,...,x_n)^T \in \mathbb{R}^n$,  is the Euclidean norm. We assume that the centers of balls in $\Lambda \cup V$ are all different (non-concentric) and that for all $(B_a,B_b) \in \Lambda^2$ and $(\mb{B}_c,\mb{B}_d) \in V^2$ (with $a \ne b$ and $c \ne d$) we have:
\begin{equation}
\label{ball_constraints}
\left\{
\begin{aligned}
&B_a \cap B_b \ne \emptyset\,, &B_a \cap B_b^{\text{c}} \ne \emptyset\,,\\       
&B_a \cap \mb{B}_c \ne \emptyset\,, &B_a \cap \mb{B}_c^{\text{c}} \ne \emptyset\,,\\
&\mb{B}_c \cap \mb{B}_d^{\text{c}} \ne \emptyset\,.\\
\end{aligned}
\right.
\end{equation}
These conditions can be verified in polynomial time in a preprocessing step in order to avoid unnecessary computations (for example, we remove $B_2$ in $\Lambda$ if $B_1 \cap B_2^{\text{c}} = \emptyset$) or trivial solutions (for example, if $B_1 \cap B_2 = \emptyset$, then $\mathcal{I}\cap \mathcal{U}^{\text{c}} = \emptyset$).
\begin{remark}
We consider open balls $B_i$. This is useful in proofs because we get the set $\mathcal{I}\cap \mathcal{U}^{\text{c}}$ open. With closed balls $B_i$ the decision problem is the same.
\end{remark}

We reformulate our geometric problem into a collection of $q$ quadratically constrained quadratic programming (QCQP) : for $j = 1,...,q$, 
\begin{equation}
\label{QCQP}
P_0(j): \,\left\{
      \begin{aligned}
       &\max_{x \in \mathbb{R}^n} (\|x-\mb{c}_{j}\|^2 - \mb{R}_j^2)\,, \\       
&\hbox{such that} \,\, \|x-c_{i}\|^2 - R_i^2 \le 0\,,\quad i =1,...,p\,,\\
&\quad\quad\quad\quad\,\,\, \|x-\mb{c}_{k}\|^2 - \mb{R}_k^2 \ge 0\,,\quad k =1,...,j-1\,.\\
      \end{aligned}
    \right.   
\end{equation}
We solve these problems sequentially for increasing integers $j$. Notwithstanding the difficult resolution of (\ref{QCQP}), the naive Algorithm \ref{algo1} gives an exact response concerning the geometric inclusion $\mathcal{I} \subset \mathcal{U}$.

\begin{algorithm}
\caption{QCQP-based naive algorithm}
\label{algo1}
\begin{algorithmic}[1]
\STATE {\bf procedure} \textsc{isIncluded}($\Lambda,V$)
   \STATE $response \gets false$, $j \gets 1$
   \WHILE{$j < q+1$}
      \STATE $x^*(j) = {\rm Argmax}\, \{P_0(j)\}$ \hfill $\triangleright$ {We solve the nonconvex QCQP problem $P_0(j)$}
      \IF{$\|x^*(j)-\mb{c}_{j}\|^2 \le \mb{R}_j^2$}\STATE $response \gets true$, $j \gets q$
      \ENDIF
   \STATE $j \gets j+1$
   \ENDWHILE
   \RETURN  $response$
\end{algorithmic}
\end{algorithm}

The feasible region $\mathcal{X}_j$ is the subset of points in $\mathbb{R}^n$ satisfying the constraints in problem $P_0(j)$. If the (nonempty) feasible region $\mathcal{X}_j$ for problem $P_0(j)$ is strictly included into $\mb{B}_j$, that is, $\|x-\mb{c}_{j}\|^2-\mb{R}_j^2 <0$ for all $x \in \mathcal{X}_j$, then $\mathcal{X}_{j+1} = \emptyset$ and $\cup_{k=1}^{j}\mb{B}_k$ covers $\mathcal{I}$. Thus, Algorithm \ref{algo1} tries to cover $\mathcal{I}$ with the balls in $V$, adding them one by one. The algorithm stops as soon as the feasible region becomes empty (justifying the while loop).

Other reformulations are possible. We choose one of them ensuring the feasibility of the considered problem at each step $j$. If feasibility is not required, solving $P_0(q)$ is enougth work. In both cases, the non-convexity of the initial geometric problem is transferred to problems $P_0(j)$ for which the feasible regions (if $q>1$ and $j>1$) and the objective functions (with the standard formulation with a minimization) are nonconvex. For this reason, line search strategies for interior point methods can fail to converge towards the global optimum. On a simple example in Figure \ref{fig:pbexample} we highlight the failing of any point following methods.
\begin{figure}[!ht]
  \begin{minipage}[c]{0.5\textwidth}
  \caption{The dark blue region shows us that $\mathcal{I} \not\subset \mathcal{U}$. Solving $P_0(1)$ stucks the solution in the green region. If only continuous and feasible moves are allowed, as with an interior point method, we would fail to solve $P_0(2)$ as its feasible region is not connected.}
\label{fig:pbexample}
  \end{minipage}\hfill
  \begin{minipage}[c]{0.4\textwidth}
  \center
 \includegraphics[totalheight=0.155\textheight]{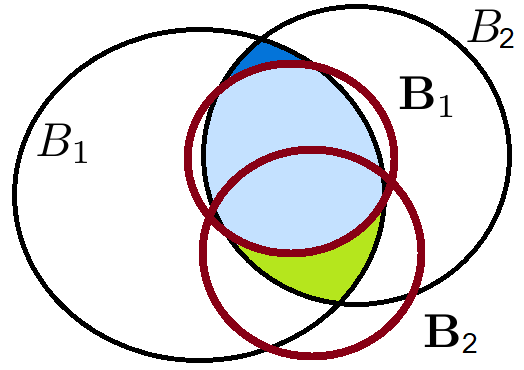} 
  \end{minipage}
  
\end{figure}

For decades, many approximate methods for the generic nonconvex QCQP problems (\ref{QCQP}) have been developed (see the review \cite{Park}) to avoid using primal methods. They are relaxation methods that usually convexify the nonconvex part of the problem or solve successive convex optimization approximate problems. Approaches as semi-definite relaxation (SDR)\cite{Luo}, reformulation linearization technique (RLT)\cite{Anstreicher} or successive convex approximation (SCA)\cite{Mehanna, Scutari} are among the most popular. However, they are often computationally greedy (using $n^2$ unknowns) and only converge towards KKT stationary points. 

We propose in this paper, to the best of our knowledge, the first exact and simple problem-solving method for the subclass of nonconvex QCQP problems involving only balls and complement of balls.

\subsection{Proposed solution}
\label{proposed}
Our QCQP problem is specific as it only involves balls. We can take advantage of the fact that the intersection of two spheres, when they meet, belongs to an hyperplane, which is both concave and convex. Considering a sphere $\mb{S}_j = \partial \mb{B}_j$, where $\partial (\cdot)$ denotes the frontier operator, we build hyperplanes as soon as a sphere $S_i = \partial B_i$ $(i=1,...,p)$ or $\mb{S}_k = \partial \mb{B}_k$ $(k \ne j,\, k =1,...,q)$ intersects $\mb{S}_j$. Each hyperplane defines a favored half-space. The intersection of the obtained half-spaces yields an open convex polyhedron $\mathtt{P}_j$ which can be seen as a Voronoi-like structure.

We introduce the notation $\mathfrak{U}_j = \cup_{k=1\,,\, k \ne j}^{q} \mb{B}_k$ and prove that we are able to detect an intersection between $\mathcal{I}$ and $\mb{S}_j \setminus \mathfrak{U}_j$ only by using the convex polyhedron $\mathtt{P}_j$. This method is based on the set equality $\mb{S}_j \cap (\mathcal{I} \setminus \mathfrak{U}_j) = \mb{S}_j \cap \mathtt{P}_j$. The closure of $\mathtt{P}_j$ is denoted $\Pi_j$ $(\Pi_j =\overline{\mathtt{P}}_j)$. Detection of a nonempty intersection $\mb{S}_j \cap \mathtt{P}_j$ can be handled solving the following $2q$ ($q$ for the minimum and $q$ for the maximum) quadratic programs (QP): for $j = 1,...,q$, 
\begin{equation}
\label{QPpolyhedron}
P_1(j): \,\left\{
      \begin{aligned}
       &{\rm extr}_{x \in \mathbb{R}^n} (\|x-\mb{c}_j\|^2 - \mb{R}_j^2)\,, \\       
&\hbox{such that} \,\,  x \in \Pi_j\,.\\
      \end{aligned}
    \right.   
\end{equation}
The minimum QP problem is tractable in polynomial time \cite{Kozlov, Ye}. Solving $P_1(j)$ for the maximum is a nonconvex (concave) problem, which can be solved by vertex enumeration in polynomial time if the polyhedron $\Pi_j$ is non-degenerate \cite{Avis}. There exists particuliar cases, in practice rarely encountered, for which the vertex enumeration problem remains NP-hard \cite{Freund}. 

In fact, as soon as we find a feasible point strictly inside the ball $\mb{B}_j$ and another one striclty outside, we have $\mb{S}_j \cap \mathtt{P}_j \ne \emptyset$ and we will prove that $\mathcal{I} \setminus \mathcal{U} \ne \emptyset$.

If for all $j$, this intersection $\mb{S}_j \cap \mathtt{P}_j$ is empty, we have $\mathcal{I} \cap ( \mb{S}_j \setminus \mathfrak{U}_j) = \emptyset$ for all $j$, that is $\mathcal{I} \cap \partial \mathcal{U} = \emptyset$. Thus, with our method we shift from the emptiness of set $\mathcal{I} \setminus \mathcal{U}$ to the emptiness of set $\mathcal{I} \cap \partial \mathcal{U}$. To solve the initial covering problem ($\mathcal{I} \setminus \mathcal{U} = \emptyset$ ?), we notice that we only need to know a point inside $\mathcal{I}$ and test whether this point is also inside $\mathcal{U}$ to decide our question.

\subsection{Outline}
\label{outline}
Section \ref{Section2} presents in detail the proposed solution. In particular, we show that the non-feasibility of Problem $P_1(j)$ gives also information on the initial geometric structure. In Appendix \ref{App1} we propose a similar method adapted to covering tests for large $q$ or sequential tests.

With some conditions on the centers and radii of the balls in $\Lambda \cup V$, we are able to ensure that the polyhedron $\Pi_j$ is unbounded and can then skip the maximization problem to only consider the convex quadratic one. We present in Section \ref{Section3} several cases where the only problem to solve is convex. If $q=1$, we introduce the so-called concave QCQP problem and highlight simplified results.

Finally, in Section \ref{Section4}, we compare our approach with recent methods of the Python library 'qcqp' gathering together the best approaches for solving nonconvex QCQP. These simulations highlight the benefit of our method specifically developed for the problem at hand.\\

This work is motivated by an application to a changepoint detection method in statistics as explained in next Subsection \ref{motivation}. This introductory section ends with a bibliographical review. 

\subsection{Motivation}
\label{motivation}

Our covering problem has a direct application for the implementation of the pruned penalized\footnote{It would also work for (non penalized) segment neighborhood method \cite{Auger}.} dynamic programming algorithm for changepoint detection in a multidimensional setting \cite{Fearnhead,Maidstone,Rigaill}. This problem consists in finding the optimal changepoint within the set $\mathbb{S}_m = \{\tau=(\tau_1,...,\tau_k) \in \mathbb{N}^k \,|\, 1 < \tau_1 < \cdots < \tau_k < m+1\}$ such that we minimize a quadratic (for Gaussian modelization) penalized cost (by $\beta>0$):
$$Q_m = \min_{\tau \in \mathbb{S}_m}\left[ \sum_{i=0}^{k}\lbrace \mathcal{C}(y_{\tau_i:\tau_{i+1}-1}) + \beta \rbrace \right],\, \hbox{with}\,\, \mathcal{C}(y_{\tau_i:\tau_{i+1}-1}) = \min_{x \in \mathbb{R}^n} \sum_{i=\tau_i}^{\tau_{i+1}-1} \sum_{k=1}^{n}(x_k- y_{ik})^2\,,$$
and $(y_1,...,y_m)^T \in (\mathbb{R}^n)^m$ the data to segment. Using a dynamic programming procedure, we build the recursion $Q_{t+1}(x) = \min \lbrace Q_{t}(x),\, m_{t}+\beta\rbrace + \sum_{k=1}^{n}(x_k- y_{(t+1)k})^2$ with $m_t = \min_{x} Q_t(x)$ and the initialization $Q_0(x) = 0$. We solve this recursion iteratively from $t=0$ to $t=m-1$.\\
At each $t$, the recursive function $Q_t(\cdot)$ is a piecewise quadratic function defined on $t$ non-overlapping regions $R^1_t,...,R^t_t \subset \mathbb{R}^n$, for which each quadratics is active on a set $R^i_t$ of type "$\mathcal{I}\setminus\mathcal{U}$". Precisely, 
$$R_t^1 = \cap_{i=1}^t B_i^1\,,\,R_t^2 = \cap_{i=2}^t B_i^2 \setminus B_1^1\,,\,R_t^3 = \cap_{i=3}^t B_i^3 \setminus (B_2^1 \cup B_2^2)\,,...,\,S_t^t = R_t^t \setminus (\cup_{j=1}^{t-1}B_{t-1}^j)\,,$$
where all the "$B_l^k$" sets designate balls determined by the data $(y_k,...,y_l)$. At the next iteration $t+1$, each $R_t^i$ is intersected by a new ball (that is, we add a ball in each $\mathcal{I}$) and the set $R_{t+1}^{t+1}$ is created. 

In order to get the global minimum $m_t$ we compare the minima of all present quadratics. To speed-up the procedure, it is worthwhile to search for vanishing sets, that is to detect efficiently the emptiness of sets $R_t^1,...,R_t^t$. In fact, once a set is proved to be empty, we do not need to consider its minimum anymore at any further iteration. This method is called {\it pruning} in the changepoint detection literature. 

In a paper in preparation, we will compare this exact method to heuristic approaches in order to build a fast algorithm for an application to genomic data. Morever, our methods can be extended to other distributions of the natural exponential family.

\subsection{Bibliographical review}

Geometric problems for balls often separately address the intersection and the union problems. Without optimization tools, the detection of a nonempty intersection between balls is difficult to solve. Helly-type theorems can be adapted to balls \cite{Baker, Maehara} but no efficient algorithm arises from this approach. The union of balls is a problem linked in the literature to molecular structures, where the volume and the surface area of molecules in 3D are important properties. Powerful algorithms based on Voronoi diagrams have been recently developed \cite{Avis1988, Cazals}. Even if the number of balls is small, that is more than two, the exact computation of simple geometric properties as volume are challenging questions \cite{Chkhartishvili}. 

One of the first problems to associate union and intersection is the historical disk covering problem, which consists in finding the minimum number of identical disks (with a given radius) needed to cover the unit disk \cite{boroczky}. This problem is still open and remains mainly unsolved, although research on this subject is active \cite{ Acharyya, Das} as it has pratical applications, for example in optical interferometry \cite{Nguyen}.

Our problem also involves covers but is different in several important ways. Indeed, we study the covering of an intersection of balls by other balls which are not necessary of the same size. Furthermore, we do not consider the question of the optimal covering, but only the covering test. This problem is part of computational geometry problems and, as far as we know, is original in the literature. Our reformulation in a QP problem plays a central role as it allows the building of an exact and efficient decision test.

Nonconvex QCQP problems are a major issue for many practical applications: problems of transmit beamforming in wireless communication \cite{Gershman} or signal processing \cite{Huang} have stimulated the development of this research area. The problem we consider in this paper is another example of a problem driven by application.

\section{Equivalent quadratic programming}
\label{Section2}

We first focus on the building of a unique polyhedron  and show its close link with the initial nonconvex sets $\mathcal{I} \setminus \mathfrak{U}_q$. Finally, we show that the $q$ polyhedra and their associated $q$ quadratic programs solve the problem.

\subsection{Linear constraints}
\label{linearconstraints}
From now one, we consider a unique problem centered on the closed ball $\mb{B}_q$ and write $\mb{B} = \mb{B}_q$, $\mb{c} = \mb{c}_q$, $\mb{R} = \mb{R}_q$, $\mathtt{P} = \mathtt{P}_j$, $\Pi = \Pi_q$, $\mathfrak{U} = \mathfrak{U}_q$. For all $i \in \{1,...,p\}$ such that $\mb{S}= \partial \mb{B}$ and $S_i = \partial B_i$ intersect, we have the hyperplane equation $h_i(x)=0$ given by
\begin{equation}
\label{hyperplane}
h_i(x)= \sum_{k=1}^n (2x_k-(\mb{c}_k+c_{ik}))(\mb{c}_k-c_{ik}) + (\mb{R}^2-R_i^2) = 0\,,
\end{equation}
and the open half-space containing the set $B_i \setminus \mb{B}$: $H_i = \{x \in \mathbb{R}^n\,,\, h_i(x) < 0\}$.

The geometric configuration of the balls $\mb{B}$ and $B_i$ with the half-space $H_i$ is given on Figure \ref{fig:balls} (left). All the balls in $\Lambda$ intersect $\mb{B}$ and the inclusion $\mb{B} \subset B_i$ is not excluded (see conditions (\ref{ball_constraints})): in this case, we do not build any hyperplane.

Similar hyperplanes and half-spaces are built between spheres $\mb{S} = \partial\mb{B}$ and $\mb{S}_j = \partial \mb{B}_j$ for $j \in \{1,...,q-1\}$ when they intersect, but here, we consider the half-space containing $\mb{B} \setminus \mb{B}_j$. Therefore,
$$\mb{h}_j(x)= \sum_{k=1}^n (2x_k-(\mb{c}_k+\mb{c}_{jk}))(\mb{c}_k-\mb{c}_{jk}) + (\mb{R}^2-\mb{R}_j^2) = 0$$
and $ \mb{H}_j = \{x \in \mathbb{R}^n\,,\, -\mb{h}_j(x)< 0\}$.

If $\mb{B} \cap \mb{B}_j = \emptyset$, a case not excluded in (\ref{ball_constraints}), we also do not build hyperplane. 
\begin{figure}[!ht]
\center
\includegraphics[totalheight=0.24\textheight]{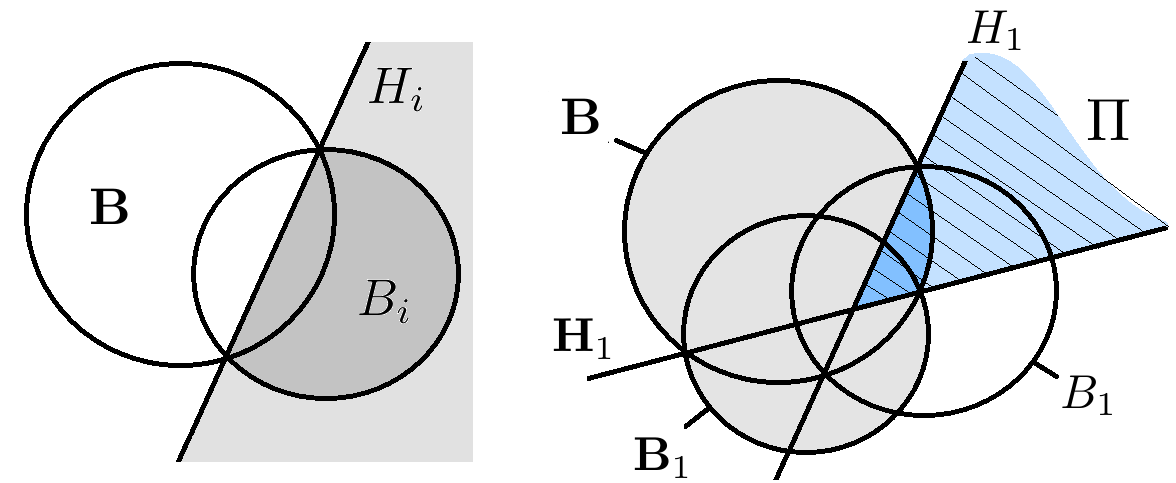} 
\caption{Left: half-space $H_i$ defined by the intersection of the spheres $\mb{S} = \partial \mb{B}$ and $S_i = \partial B_i$ and containing $B_i \setminus \mb{B}$. Right: example of feasible region $\Pi$ (the hatched area) with $p=1$ and $q=2$. Here, $\mathtt{P}^+ = H_1$ and $\mathtt{P}^- = \mb{H}_1$. The darker shade of blue is the subset of $\Pi$ lying inside the ball $\mb{B}$.}
\label{fig:balls}
\end{figure}

With these half-spaces, we define the open convex polyhedron
$$\mathtt{P} = \{x \in \mathbb{R}^n\,,\, h_i(x) < 0\,,\, -\mb{h}_j(x) < 0\,,\, i=1,...,p\,,\, j=1,...,q-1\}$$ 
and the polyhedra $\mathtt{P}^+ = \cap_{i=1}^p H_i$ and $\mathtt{P}^- = \cap_{j=1}^{q-1} \mb{H}_j$ such that $\mathtt{P} = \mathtt{P}^+ \cap \mathtt{P}^-$. The open polyhedron $\mathtt{P}$ will be used to prove geometrical properties, whereas its closure $\Pi = \overline{\mathtt{P}}= \overline{\mathtt{P}^+} \cap \overline{\mathtt{P}^-} = \Pi^+\cap \Pi^-$ is the feasible region of the QP problem $P_1(q)$.

An efficient resolution of the nonconvex problem of type $P_1$ (cf (\ref{QPpolyhedron})) for maximum is made possible, solving for example a vertex enumeration problem \cite{Avis}. An example of feasible region is drawn on Figure \ref{fig:balls} (right).\\

Before proving the equivalence of Problems (\ref{QCQP}) and (\ref{QPpolyhedron}), we present some simple equalities and inclusions used throughout this paper between sets involving balls and hyperplanes.
\begin{lemma}
\label{lemmaI}
For $i \in \{1,...,p\}$ and $j \in \{1,...,q-1\}$, we have the relations: 
\begin{align*} 
&(a)\quad \mb{S}\cap B_i = \mb{S}\cap H_i\,,&
&(b)\quad \mb{S}\cap \mb{B}_j^{\text{c}} = \mb{S}\cap \mb{H}_j\,,\\
&(c)\quad B_i \setminus \mb{B} \subset H_i\,,&
&(d)\quad \mb{B} \setminus \mb{B}_j \subset \mb{H}_j\,.
\end{align*}
\end{lemma}
The proof is straightforward.
 
\subsection{Links between the polyhedron and the initial geometric problem}

We give a set equality linking the convex polyhedron $\mathtt{P}$ and the open set $\mathcal{I} \setminus \mathfrak{U}$ involved in the initial geometric problem.

\begin{corollary} \label{equality}$\mb{S}\cap (\mathcal{I} \setminus \mathfrak{U}) = \mb{S}\cap \mathtt{P}$.
\end{corollary}

\begin{proof}
$\mb{S}\cap \mathcal{I} = \mb{S}\cap( \cap_{i=1}^pB_i) = \cap_{i=1}^p (\mb{S}\cap B_i)=\cap_{i=1}^p (\mb{S}\cap H_i) = \mb{S}\cap(\cap_{i=1}^p H_i)= \mb{S}\cap \mathtt{P}^+$ by relation (a) of Lemma \ref{lemmaI}. In the same way, we have $\mb{S}\cap \mathfrak{U}^{\text{c}} = \mb{S}\cap (\cup_{j=1}^{q-1} \mb{B}_j)^{\text{c}} = \mb{S}\cap (\cap_{j=1}^{q-1} \mb{B}_j^{\text{c}}) = \cap_{j=1}^{q-1}(\mb{S} \cap \mb{B}_j^{\text{c}})=\cap_{j=1}^{q-1}(\mb{S} \cap \mb{H}_j)= \mb{S}\cap \mathtt{P}^-$ by De Morgan's laws and relation (b) of Lemma \ref{lemmaI}. Therefore, $\mb{S}\cap \mathcal{I} \setminus \mathfrak{U} = (\mb{S}\cap \mathcal{I})\cap (\mb{S}\cap\mathfrak{U}^{\text{c}}) = (\mb{S}\cap \mathtt{P}^+)\cap (\mb{S}\cap  \mathtt{P}^-) = \mb{S}\cap \mathtt{P}$ and the relation is proven.
\end{proof}

This corollary is a central result. The following proposition is a technical topological result used in the next subsection.

For $\epsilon \in \mathbb{R}$, we introduce the ball $\mb{B}^{R+\epsilon}$ with center $\mb{c}$ and radius $\mb{R}+\epsilon$. We also define $\mb{S}^{R+\epsilon} = \partial \mb{B}^{R+\epsilon}$. $V_n$ is the volume (Lebesgue measure) in dimension $n$. In the following proposition, we connect the position of points in the closed feasible region $\Pi$ to the spatial position of $\mathcal{I} \setminus \mathfrak{U}$.

\begin{proposition}
\label{fondprop}
If the volume of the feasible region is nonzero  ($V_n(\Pi)>0$), the following assertions are equivalent:\\
i) $\mb{S}\cap \mathtt{P} \ne \emptyset$,\\
ii) there exists $(x^-, x^+)\in \Pi \times \Pi$ such that $\|x^--\mb{c}\|^2 < \mb{R}^2 < \|x^+-\mb{c}\|^2$,\\
iii) $V_{n-1}(\mb{S} \cap \Pi) >0$ (the Lebesgue measure of the surface area is nonzero),\\
iv) there exists $r > 0$ such that for all $\epsilon \in (-r,r)$, $\mb{S}^{R+\epsilon} \cap (\mathcal{I} \setminus \mathfrak{U}) \ne \emptyset$.
\end{proposition}

To state this result, we use the following lemma. 
\begin{lemma}~\\
\label{openset}
With $\mathcal{O}$ an open set of $\mathbb{R}^n$, we have equivalence between propositions:\\
a) $\mb{S} \cap \mathcal{O} \ne \emptyset$;\\
b) $V_{n-1}(\mb{S} \cap \mathcal{O}) >0$;\\
c) there exists $r > 0$ such that for all $\epsilon \in (-r,r)$, $\mb{S}^{R+\epsilon} \cap \mathcal{O} \ne \emptyset$.
\end{lemma}
\begin{proof}
With $\tilde{x} \in \mb{S} \cap \mathcal{O}$, there exists an open ball $B(\tilde{x},\rho)$ centered on $\tilde{x}$ and of radius $\rho$ such that $B(\tilde{x},\rho) \subset \mathcal{O}$. Thus $V_{n-1}(\mb{S} \cap \mathcal{O}) \ge V_{n-1}(\mb{S} \cap B(\tilde{x},\rho)) >0$ and for all $\epsilon \in (-\rho,\rho)$, $\mb{S}^{R+\epsilon} \cap \mathcal{O}\subset \mb{S}^{R+\epsilon} \cap B(\tilde{x},\rho) \ne \emptyset$. We have shown that $a) \implies b) \implies c)$. The implication $c) \implies a)$ is obvious and the Lemma is proven. 
\end{proof}

\begin{proof}[Proof of Proposition \ref{fondprop}]
The implication $i) \implies ii)$ is due to the openness of $\mathtt{P}$. We prove the converse. As $V_n(\Pi)>0$ and $\Pi$ is convex, we have $\mathtt{P} \ne \emptyset$ and $\overline{\mathtt{P}}=\Pi$. Thus, there exist sequences $(x_n^-)$ and $(x_n^+)$ such that $\lim_{n \to +\infty} x_n^- = x^-$ and $\lim_{n \to +\infty} x_n^+ = x^+$ with $x_n^-,x_n^+ \in \mathtt{P}$ for all $n \in \mathbb{N}$. Therefore, there exists $N\in \mathbb{N}$ such that, for all $n \ge N$, $\|x_n^--\mb{c}\|^2 < \mb{R}^2 < \|x_n^+-\mb{c}\|^2$. By convexity of $\mathtt{P}$, we get $\mb{S} \cap \mathtt{P} \ne \emptyset$. We use the Lemma \ref{openset} ($a \iff b$) with the open set $\mathtt{P}$ to prove the equivalence between $i)$ and $iii)$, knowing that $V_{n-1}(\mb{S} \cap \Pi)= V_{n-1}(\mb{S} \cap \mathtt{P})$. As $\mb{S}\cap (\mathcal{I} \setminus \mathfrak{U}) = \mb{S}\cap \mathtt{P}$ (see Corollary \ref{equality}), we use again Lemma \ref{openset} ($b \iff c$) with the open set $\mathcal{I} \setminus \mathfrak{U}$ to get the equivalence between propositions $iii)$ and $iv)$.
\end{proof}

With this last result, we have shown that some propositions involving the closed polyhedron $\Pi$ are related to results using $\mathtt{P}$ and hence $\mathcal{I} \setminus \mathfrak{U}$.  We now combine this result with the $q$ quadratic programming problem to solve the problem.

\subsection{Quadratic programming decision}
We still focus on a unique quadratic program:
\begin{equation}
\label{QPextr}
P_1(q) = P_1: \,\left\{
      \begin{aligned}
       &{\rm extr}_{x \in \mathbb{R}^n} (\|x-\mb{c}\|^2 - \mb{R}^2)\,, \\       
&\hbox{such that} \,\,  x \in \Pi\,,\\
      \end{aligned}
    \right.   
\end{equation}
and we denote by $x^*_{min}$ (resp. $x^*_{max}$) the value of the argument $x$ for which the objective function attains its minimum (resp. maximum) over the set $\Pi$. Notice that, if $\Pi$ is unbounded, $x^*_{max} = +\infty$ and in some singular cases, the argument of the maximum may not be unique.

The non-existence of a point $x^-$ or $x^+$ satisfying the strict inequalities in Proposition \ref{fondprop} case (i) is related to the resolution of problem $P_1$. For example, there is no point $x^-$ if $\mb{R}^2 \le \|x^*_{min}-\mb{c}\|^2$. Before solving this extremum problem, one should verify the feasibility of the constraints. Studying the feasibility of these constraints, we get some inclusion criteria.

\begin{proposition}~\\
\label{infeasible}
If $V_n(\Pi^+) = 0$, then $\mathcal{I} \subset \mb{B}$ (or $\mathcal{I} = \emptyset$),\\
if $V_n(\Pi^-) = 0$, then $\mb{B} \subset \mathfrak{U}$,\\
if $V_n(\Pi) = 0$, then $\mathcal{I}\cap (\mb{S} \setminus \mathfrak{U}) = \emptyset$.
\end{proposition}

\begin{proof}
$\mathcal{I} \setminus \mb{B} = \cap_{i=1}^{p}(B_i \setminus \mb{B}) \subset \cap_{i=1}^{p} H_i  = \mathtt{P}^+$ using relation (c) of Lemma \ref{lemmaI}. If $V_n(\Pi^+) = 0$, then $\mathtt{P}^+ = \emptyset$ and we have $\mathcal{I} \subset \mb{B}$. $\mb{B} \setminus \mathfrak{U} = \cap_{j=1}^{q-1}(\mb{B} \setminus \mb{B}_j) \subset \cap_{j=1}^{q-1} \mb{H}_j  = \mathtt{P}^- $ using relation (d) of Lemma \ref{lemmaI}. If $V_n(\Pi^-) = 0$, then $\mathtt{P}^- = \emptyset$ and we have $\mb{B} \subset \mathfrak{U}$. The last result is a direct application of Corollary \ref{equality}.
\end{proof}

We can now present our main result, showing the equivalent between the resolution of the QP problems $P_1(j)$, $j=1,...,q$, and the decision about $\mathcal{I} \setminus \mathcal{U}$.

\begin{theorem}~\\
\label{theorem}
(A) If there exists $(x^-, x^+)\in \Pi \times \Pi$ such that $\|x^--\mb{c}\|^2 < \mb{R}^2 < \|x^+-\mb{c}\|^2$, then $\mathcal{I} \setminus \mathcal{U} \ne \emptyset$, $V_n(\mathcal{I} \setminus \mathcal{U})>0$ and $V_n(\mathcal{I} \cap \mathcal{U})>0$;\\
(B) if $\mb{R}^2 \le \|x^*_{min}-\mb{c}\|^2$ or $\|x^*_{max}-\mb{c}\|^2 \le \mb{R}^2$, then  $\mathcal{I} \cap (\mb{S}\setminus \mathfrak{U}) = \emptyset$.
\end{theorem}

\begin{proof}
In case (A), Proposition \ref{fondprop} gives the existence of $r > 0$ and $x \in \mathbb{R}^n$ such that $x \in \mb{S}^{R+\frac{r}{2}} \cap (\mathcal{I} \setminus \mathfrak{U})$. In particular, $x \in \mb{S}^{R+\frac{r}{2}}$, then $x \in \mb{B}^{\text{c}}$. Consequently, $x \in (\mathcal{I} \setminus \mathfrak{U}) \cap \mb{B}^{\text{c}} = \mathcal{I} \setminus \mathcal{U}$ (an open set) and $V_n(\mathcal{I} \setminus \mathcal{U}) > 0$. With a point $x' \in \mb{S}^{R-\frac{r}{2}}\cap (\mathcal{I} \setminus \mathfrak{U})$ we get $\mb{B} \cap \mathcal{I} \ne \emptyset$ and $V_n(\mathcal{I} \cap \mathcal{U}) > 0$. Case (B) is the negation of case (A). That is, there exists $r>0$ such that for all $\epsilon \in (0,r)$ or for all $\epsilon \in (-r,0)$, we have $\mb{S}^{R+\epsilon} \cap (\mathcal{I} \setminus \mathfrak{U}) = \emptyset$ using Proposition \ref{fondprop}. Therefore $\mathcal{I} \cap (\mb{S}\setminus \mathfrak{U}) = \emptyset$. Otherwise, with $\mathcal{O} = \mathcal{I} \setminus \mathfrak{U}$ in Lemma \ref{openset} we get  $\mb{S}^{R+\epsilon} \cap (\mathcal{I} \setminus \mathfrak{U}) \ne \emptyset$ for all $\epsilon \in (-r^*,r^*)$ and fixed $r^* > 0$, which is impossible.
\end{proof}

If case (B) is satisfied, the convex set $\mathcal{I}$ does not intersect the frontier $\mb{S}\setminus \mathfrak{U}$ and we consider another reference ball in $V$ to solve a new $P_1$-type problem (\ref{QPpolyhedron}). If for all elements of $V$ we get case (B), we have $\mathcal{I} \cap \partial \mathcal{U} = \emptyset$. Using a point $\hat{x}$ in $\mathcal{I}$, we test whether $\hat{x}$ is included in $\mathcal{U}$ to conclude.

\section{Simplifications}
\label{Section3}

\subsection{Convex reduction}
\label{subsection31}
We show that if some mild conditions are satisfied, the feasible region $\Pi$ of Problem $P_2$ (see (\ref{QPextr})) can be unbounded and the maximum problem in Theorem \ref{theorem} does not have to be solved anymore ($\|x^*_{max}-\mb{c}\| = +\infty$). The problem is reduced to a convex quadratic programming problem which can be solved in polynomial time \cite{Kozlov, Ye}. 
\begin{theorem}
\label{theorem2}
For all $i \in \{1,...,p\}$ and  $j \in \{1,...,q-1\}$, if $R_i^2 < \mb{R}^2+\|\mb{c}-c_i\|^2$ and $\mb{R}_j^2 > \mb{R}^2+\|\mb{c}-\mb{c}_j\|^2$, then $\Pi \setminus \mb{B} \ne \emptyset$ or $\Pi = \emptyset$.
\end{theorem}

\begin{proof}
First, we show that $\mb{c} \not\in \Pi$. Indeed, we have for $i = 1,...,p$,
$$h_i(\mb{c})= \sum_{k=1}^n (2\mb{c}_k-(\mb{c}_k+c_{ik}))(\mb{c}_k-c_{ik}) + (\mb{R}^2-R_i^2) = \mb{R}^2 + \|\mb{c}- c_i\|^2 - R_i^2 > 0\,,$$
and for $j = 1,...,q-1$,
$$-\mb{h}_j(\mb{c})= \sum_{k=1}^n -(2\mb{c}_k-(\mb{c}_k+\mb{c}_{jk}))(\mb{c}_k-\mb{c}_{jk}) - (\mb{R}^2-\mb{R}_j^2) = \mb{R}_j^2 - \mb{R}^2 - \|\mb{c}- \mb{c}_j\|^2 > 0\,.$$ 
using the hypothesis. Second, we prove that, if $\Pi \ne \emptyset$, then $\Pi$ is unbounded. Suppose that there exists $\hat{x} \in \Pi$. We consider the linear functions $g_i : \lambda \mapsto h_i(\hat{x}+\lambda(\mb{c}-\hat{x}))$ and $\mb{g}_i : \lambda \mapsto -\mb{h}_i(\hat{x}+\lambda(\mb{c}-\hat{x}))$ with $\lambda \in \mathbb{R}$. We have $g_i(0)\le 0$, $g_i(1)> 0$ and $\mb{g}_i(0)\le 0$, $\mb{g}_i(1)> 0$, so that these functions are strictly increasing. Thus, for all negative lambda, we have $\hat{x}+\lambda(\mb{c}-\hat{x}) \in \Pi$. Therefore, as $\lim_{\lambda \to -\infty} \|\hat{x}+\lambda(\mb{c}-\hat{x})\| = +\infty$, the set $\Pi \setminus \mb{B}$ can not be empty.
\end{proof}

With only an upper bound on the number of involved balls in $\Lambda \cup V$, we have the same result.

\begin{proposition}
\label{dimensionCondition}
If the number of constraints is less or equal to the dimension, that is $p+q-1 \le n$, then $\Pi$ is unbounded (or empty) and the conclusion of Theorem \ref{theorem2} still holds.
\end{proposition}

\begin{proof}
We explicit the constraints as a system of linear inequalities, $Ax \le b$, with $A \in \mathbb{R}^{(p+q-1)\times n}$ and $b \in \mathbb{R}^{p+q-1}$. If $A$ is nonsingular, the system $Ax = b - t \mathbb{I}$ with $\mathbb{I}$ the vector filled by ones, has a (unique) solution $x(t)$ and $\lim_{t\to +\infty}\|x(t)\| = \lim_{t\to +\infty}\|A^{-1}(b - t \mathbb{I})\| = +\infty$ with $Ax(t) = b - t \mathbb{I} \le b$. If $A$ is singular, the rank-nullity theorem shows that the linear subspace $Ker(A)$ is of dimension $n-rk(A)>0$ and then $Ker(A) \ne \{0\}$. If a point $x_0 \in \mathbb{R}^n$ satisfies the inequalities, thus, $A(x_0+ty) = Ax_0 \le b$ with $y \in Ker(A) \setminus \{0\}$ and $t\in \mathbb{R}$. We have $\lim_{t\to +\infty}\|x_0+ty\|= +\infty$ and the result is proven.
\end{proof}

\begin{remark}
The conditions of Lemma \ref{theorem2} are sharp. If one of the relations is false, then there exist sets of balls $\Lambda$ and $V$ with $\#(\Lambda\cup V)= p+q = n+2$, such that we have $\Pi \subset \mb{B}$. An example of such a set $\Pi$ is the $n$-dimensional pyramid with summit near point $\mb{c}$ and a basis  obtained with the only one hyperplane that do not satisfies the conditions of the theorem. The necessary condition of Proposition \ref{dimensionCondition} is verified as $p+q > n +1$.
\end{remark}
\begin{theorem}
If there exists $d \in \mathbb{R}^n \setminus \{0\}^n$ such that $(\mb{c}-c_{i})^T d > 0$ and $(\mb{c}-\mb{c}_{j})^T d <0 $ for all $i \in \{1,...,p\}$ and $j \in \{1,...,q-1\}$ (linear separability) then $\Pi$ is unbounded or $\Pi = \emptyset$.
\end{theorem}

\begin{proof}
Suppose that there exists $\hat{x} \in \Pi$ and $d$ satisfying the condition of the theorem. Then for all $\alpha < 0$, we have for all the constraints
$$h_i(\hat{x} + \alpha d)= (2\hat{x}-(\mb{c}+c_{i}))^T(\mb{c}-c_{i}) + 2\alpha d^T(\mb{c}-c_{i}) + (\mb{R}^2-R_i^2) \le  h_i(\hat{x}) \le 0\,,$$
$$-\mb{h}_j(\hat{x} + \alpha d)= (2\hat{x}-(\mb{c}+\mb{c}_j))^T(\mb{c}-\mb{c}_j) - 2 \alpha d^T(\mb{c}-\mb{c}_j) - (\mb{R}^2-\mb{R}_j^2) \le  -\mb{h}_j(\hat{x}) \le 0\,$$ 
and $\hat{x} + \alpha d$ definies an infinite direction inside the polyhedron $\Pi$.
\end{proof}

\subsection{The concave QCQP problem}
If there is only one ball in the set $V$ and then no more concave constraint in (\ref{QCQP}), we consider only one QCQP problem of type (\ref{QCQP}) and only one QP problem of type (\ref{QPextr}) ($q=1$). We say that we solve a concave QCQP problem because we minimize the opposite of a convex function over a set of convex constraints.

A first direct consequence is that $\Pi = \Pi^+$, $\Pi^- = \mathbb{R}^n$ and $\mathfrak{U}=\emptyset$. In this particular configuration, some of our results in Section \ref{Section2} and Subsection \ref{subsection31} can be simplified.\\

\texttt{Feasibility}:\\
If $V_n(\Pi) = 0$, then $\mathcal{I} \subset \mb{B}$ (or $\mathcal{I} = \emptyset$).\\

\texttt{Quadratic programming reduction}:\\
(A) If $(x^-, x^+)\in \Pi \times \Pi$ is such that $\|x^--\mb{c}|\|^2 < \mb{R}^2 < \|x^+-\mb{c}\|^2$, then $\mathcal{I} \setminus \mb{B} \ne \emptyset$;\\
(B) if $\mb{R}^2 \le \|x^*_{min}-\mb{c}\|^2$, then  $\mathcal{I} \cap \mb{S} = \emptyset$;\\
(C) if $\|x^*_{max}-\mb{c}\|^2 \le \mb{R}^2$ then $\mathcal{I} \subset \mb{B}$.

\begin{proof}
Only the case (C) has to be proven. $\mathcal{I} \setminus \mb{B} = \cap_{i=1}^p (B_i \setminus \mb{B}) \subset \cap_{i=1}^p (H_i \setminus \mb{B}) = \Pi \setminus \mb{B}$ using Lemma \ref{lemmaI}. However, $\Pi \setminus \mb{B} = \emptyset$ and then $\mathcal{I} \subset \mb{B}$.
\end{proof}

\texttt{Convex optimization reduction}:\\
if $R_i^2 < \mb{R}^2+\|\mb{c}-c_i\|^2$, for all $i \in \{1,...,p\}$, then $\Pi \setminus \mb{B} \ne \emptyset$ or $\Pi = \emptyset$.\\

\section{Simulations}
\label{Section4}
Algorithm \ref{algo2} describes our new procedure based on equivalent QP problems. It is made of two steps: the detection of an intersection between $\mathcal{I}$ and $\partial \mathcal{U}$ using QP problems (see Section \ref{Section2}) and the test $\mathcal{I} \subset \mathcal{U}$ with a unique point (if $\mathcal{I} \cap \partial \mathcal{U} = \emptyset$). We chose to determine $x^*_{min}$ and $x^*_{max}$ for each QP problem but in fact we only need points $x^-$ and $x^+$ as shown by Theorem \ref{theorem}. 

\begin{algorithm}
\caption{QP-based decision algorithm}
\label{algo2}
\begin{algorithmic}[1]
\STATE {\bf procedure} \textsc{IsIncluded}($\Lambda,V$)
   \STATE $response \gets true$, $j \gets 1$
   \WHILE{$j < q+1$}
      \STATE $x^*_{min}(j) = {\rm Argmin}\, \{P_1(j)\}$ \hfill $\triangleright$ {We solve the QP problem $P_1(j)$ for min} 
            \IF{$\|x^*_{min}(j)-\mb{c}_{j}\|^2 < \mb{R}_j^2$}    
	  \STATE $x^*_{max}(j) = {\rm Argmax}\, \{P_1(j)\}$ \hfill $\triangleright$ {We solve the QP problem $P_1(j)$ for max}
      \IF{$\mb{R}_j^2 < \|x^*_{max}(j)-\mb{c}_{j}\|^2$}\STATE $response \gets false$, $j \gets q$
       \ENDIF
      \ENDIF
   \STATE $j \gets j+1$
   \ENDWHILE
    \IF{$response = true$}\STATE Find $\tilde{x} \in \mathcal{I}$
    \IF{$\tilde{x} \not\in \mathcal{U}$}\STATE $response \gets false$ \ENDIF
         \ENDIF
   \RETURN  $response$
\end{algorithmic}
\end{algorithm}

Our goal is to compare the performances of this algorithm against Algorithm \ref{algo1} on two aspects: the exactness of the result and their computational efficiency. We will also illustrate the polynomial time complexity of our new algorithm by increasing the dimension $n$ at fixed $p$ and $q$.\\ 

We have implemented\footnote{The Python code is available on my github 'vrunge' in repository called 'QP\_vs\_QCQP'
 \url{https://github.com/vrunge/QP_vs_QCQP/blob/master/simulations.ipynb}} the Algorithm \ref{algo1} in Python using the recent (2017) {\it suggest and improve} method for nonconvex QCQP \cite{Park} based on packages 'cvxpy' \cite{cvxpy} and its extension 'qcqp'. This allows us to use state-of-the-art methods on Python, more elaborated than interior point one, which is not adapted to our problem (see Figure \ref{fig:pbexample}). The suggest method is set to {\it random} and the improve step to {\it ADMM} (alternating directions method of multipliers \cite{Boyd}) followed by a coordinate descent step. This is the most efficient combinaison of methods (empirically speaking) and is used in some examples given in the package. 
 
The center points for balls are randomly generated with a normal distribution centered on zero with standard deviation $\sigma = 10$ and their radius is the distance to zero plus a quantity $\epsilon = 5$ for intersection ball and $2\epsilon$ for union balls. We also verify that the obtained balls intersect each other as in (\ref{ball_constraints}), if not, we simulate new balls. Examples of simulated data are shown in Figure \ref{fig:dataSimu}.\\

\begin{figure}[!ht]
\center
\includegraphics[totalheight=0.18\textheight]{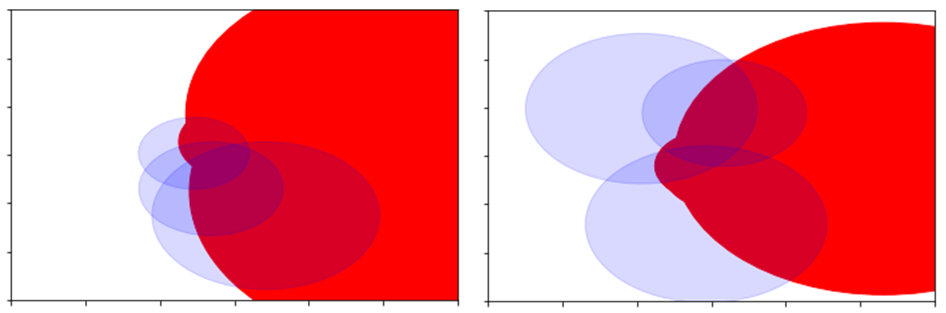} 
\caption{Two examples of simulated data with $p = q = 3$ in dimension $2$. On the left, $\mathcal{I} \not\subset \mathcal{U}$, on the rigtht $\mathcal{I} \subset \mathcal{U}$. The opaque red disks are the balls of the union set $V$. The transparent blue disks are the balls of the intersection set $\Lambda$.}
\label{fig:dataSimu}
\end{figure}

In Table \ref{table1}, we generate 100 examples for each proposed configuration. We first consider examples with $\mathcal{I} \not\subset \mathcal{U}$. In case $(n=2, p =3)$ we compare the results of the two algorithms with the 2d graphical representation of the balls problem to confirm the exactness of our method. 

Algorithm \ref{algo2} always finds the exact result in dimension $2$ and we get the value $100$. For higher dimensions, we know the result of Algorithm \ref{algo2} to be exact (as shown by theoretical results of Section \ref{Section2}) and count the number of identical results between the two algorithms to get the number of true results for Algorithm \ref{algo1}. Algorithm \ref{algo1} often fails, in particular when the non inclusion is difficult to visually detect (obtained on a small volume). This behavior is a consequence of our simulations: we generated data with a unique simulation procedure, so that, for increasing $q$ it becomes harder to generate examples without inclusion. An increasing dimension $n$ gives more space to get a non inclusion and facilitate the detection of $\mathcal{I} \not\subset \mathcal{U}$ for Algorithm \ref{algo1}.

In the case of the inclusion $\mathcal{I} \subset \mathcal{U}$, we notice that Algorithm \ref{algo1} always finds the right answer in our simulations. Interestingly, the resolution of the unique problem $P_0(q)$ often fails, unlike the non inclusion case.

\begin{table}[tbhp]
{
\caption{Number of exact results over 100 simulations}\label{table1}
\begin{center}

\begin{tabular}{|c|c|c|c|c|c|c|c|} \hline
$\mathcal{I} \not\subset \mathcal{U}$ & \multicolumn{3}{ c |}{$n=2$, $p=3$}&\multicolumn{3}{ c |}{$p=3$, $q=3$}& $p=5$, $q=5$  \\ 
 & $q = 1$ & $q = 2$ & $q = 3$&$n = 3$&$n = 5$&$n = 10$&$n = 10$ \\ \hline
Algorithm \ref{algo1} & 83 & 65 & 48 & 57 & 63& 79 & 65 \\
$P_0(q)$ & 83 & 80 & 87 & 91 & 91 & 90 & 92 \\
Algorithm \ref{algo2} & 100 & 100 & 100&$\cdot$&$\cdot$&$\cdot$&$\cdot$ \\ \hline
$\mathcal{I} \subset \mathcal{U}$ &\multicolumn{7}{ c |}{}\\ 
 \hline
Algorithm \ref{algo1} & 100 & 100 & 100 & 100 & 100 & 100 & 100 \\
$P_0(q)$ & 100 & 92 & 81 & 75 & 75 & 83 & 77\\
\hline
\end{tabular}
\end{center}
}
\end{table}

In Table \ref{table2} we highlight the efficiency of our new algorithm compared with the Python package 'qcqp' based on nonconvex QCQP methods. Notice that improvements in the code for Algorithm \ref{algo2} are still possible with a direct implementation of the QP problem stopping as soon as we get $x^-$ or $x^+$ and in the vertex enumeration solver (we simply used package cdd\footnote{website of the Python package \url{http://pycddlib.readthedocs.io/en/latest/}}\cite{Fukuda}). We have made $10$ simulations ($i= 1,...,10$) for $3$ configurations ($n = 5, 10, 20$) with $p=q=3$ in case $\mathcal{I} \not\subset \mathcal{U}$ (simplier to obtain in our ball generation algorithm with increasing dimension $n$). Our results show for the Algorithm \ref{algo1} a computational time complexity of order $O(10^3)$ slower than for Algorithm \ref{algo2} and a higher time variability between iterates.

\begin{table}[tbhp]
{
\caption{Comparison of execution time with $p=q=3$ and varying dimension $n$}\label{table2}
\begin{center}
\begin{tabular}{|c|c|c|c|c|c|c|} \hline
& \multicolumn{2}{ c |}{$n=5$}&\multicolumn{2}{ c |}{$n=10$}& \multicolumn{2}{ c |}{$n=20$}  \\ 
i & time QP & time QCQP & time QP & time QCQP&time QP & time QCQP \\ \hline
1&0.00411&2.013 &0.00448&2.667 &0.00464& 2.923\\
2&0.00280&1.213 &0.00261&2.553 &0.00298& 1.732\\
3&0.00260&1.034 &0.00262&2.358 &0.00305& 2.570\\
4&0.00274&1.044 &0.00262&2.133 &0.00290& 15.10\\
5&0.00279&1.967 &0.00272&2.098 &0.00306& 4.843\\
6&0.00304&1.039 &0.00278&1.402 &0.00294& 5.856\\
7&0.00267&1.639 &0.00267&1.766 &0.00288& 48.80\\
8&0.00261&1.070 &0.00277&8.500 &0.00296& 2.780\\
9&0.00547&0.832 &0.00275&1.803 &0.00375& 2.913\\
10&0.00263&0.564 &0.00286&1.663&0.00666& 2.964\\
\hline
& \multicolumn{6}{ c |}{RATIO mean time QCQP by mean time QP} \\ 
& \multicolumn{2}{ c |}{394.6}&\multicolumn{2}{ c |}{933.0}& \multicolumn{2}{ c |}{2526}  \\ 
\hline
\end{tabular}
\end{center}
}
\end{table}

An empirical polynomial time complexity is confirmed by our simulations as shown in Figure \ref{fig:time}. Simulations have been performed with $\mathcal{I} \not\subset \mathcal{U}$ for an increasing $n$ so that we only solve a unique QP problem for maximum and minimum to detect an intersection. With a linear regression based on the model $t = an^b$, we get a power of about $b = 1.76$ with $R^2 = 0.935$ in the logarithmic scale.

\begin{figure}[!ht]
\center
\includegraphics[totalheight=0.22\textheight]{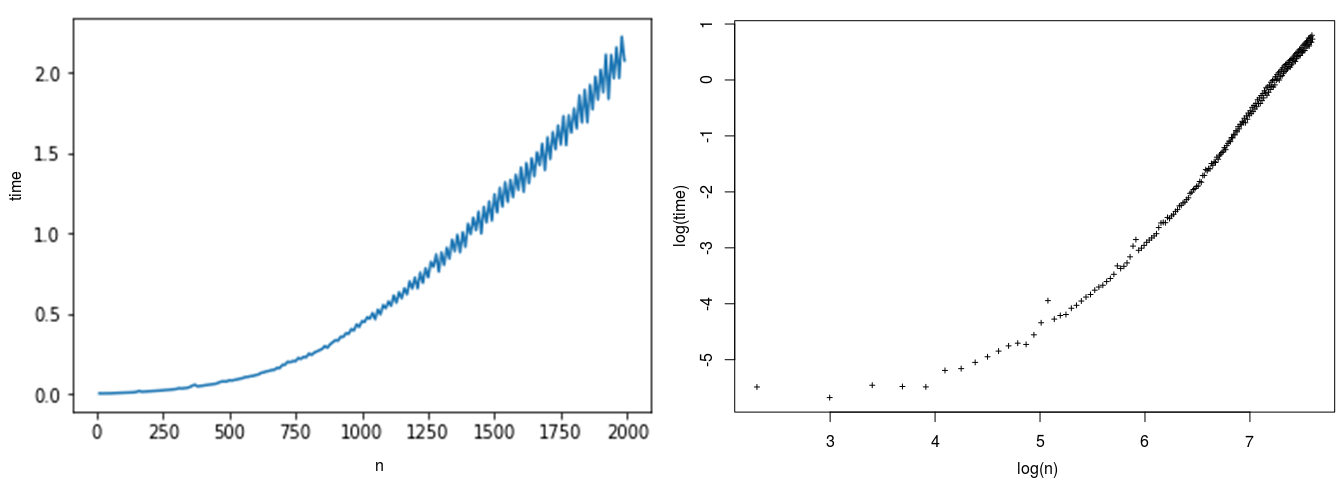} 
\caption{Left: $p = q= 3$ and $n = 10,20,30,...,2000$. Each datapoint is the mean over $10$ identical simulations in case $\mathcal{I} \not\subset \mathcal{U}$. A unique QP problem is solved for each simulation. Right: In the log-log scale, the curve stabilizes corroborating a polynomial time complexity of type $an^b$. A linear model return a coefficient of determination of $0.935$ with $b \approx 1.76$.}
\label{fig:time}
\end{figure}

\section*{Conclusion}
\label{conclusion}
The geometric question of the cover of an intersection of balls by an union of other balls was addressed using optimization tools. The collection of nonconvex quadratically constrained quadratic programming problems has been transformed into a collection of quadratic optimization problems: the minimum and the maximum distances between a point and a Voronoi-like polyhedron has to be found for each problem. The maximum problem can be handled efficiently by vertex enumeration. If simple conditions are satisfied, the polyhedron is unbounded, the maximization problem does not have to be considered anymore and the complexity is known to be in polynomial time. Simulations show that state-of-the-art nonconvex QCQP algorithms developed in Python often fail to answer our question and are computationally greedy. Our new method never fails (is exact) and efficiently find the solution with a time complexity of order $O(10^3)$ smaller in our simulation study. In a further work, we will apply this method to the efficient implementation of a multidimensional changepoint detection algorithm based on pruned dynamic programming.
\appendix
\section{An alternative method}
\label{App1}
In some applications, as for example in Subsection \ref{motivation}, we need to test sequentially the covering. This means that we add at each iteration a new ball in the intersection $\mathcal{I}$. For large $q$, it could be expensive to solve $q$ QP problems at each iteration for minimum and maximum. Another approach consists in detecting an intersection between $\mathcal{I}\setminus \mathcal{U}$ and $\mathbb{S}$ (as soon as $\#\mathcal{I} >0$), where $\mathbb{S}$ is the frontier of the new ball $\mathbb{B}$ (with center $c$ and radius $R$) to add in $\mathcal{I}$. The problem is now centered on the "newest" ball of set $\Lambda$ rather than the successive balls of $V$. Thus, a unique QP problem centered on the ball $\mathbb{B}$ is built and we get a result similar to Theorem \ref{theorem}:

\begin{theorem}~\\
(A) If there exists $(x^-, x^+)\in \Pi \times \Pi$ such that $\|x^--c\|^2 < R^2 < \|x^+-c\|^2$, then $(\mathbb{B} \cap \mathcal{I}) \setminus \mathcal{U} \ne \emptyset$ and $V_n((\mathbb{B} \cap \mathcal{I}) \setminus \mathcal{U}) > 0$;\\
(B) if $R^2 \le \|x^*_{min}-c\|^2$ or $\|x^*_{max}-c\|^2 \le R^2$, then  $\mathbb{S}\cap (\mathcal{I}\setminus \mathcal{U}) = \emptyset$.
\end{theorem}

In case (B) we can conclude knowing a point $\tilde{x}_k$ for each connected component $C_k$ $(k=1,...,K)$ of $\mathcal{I}\setminus \mathcal{U}$ and testing $\tilde{x}_k \in \mathbb{B}$. If always false, we get $(\mathbb{B} \cap \mathcal{I}) \setminus \mathcal{U} = \emptyset$. The drawback of this approach is the necessity to know the number of connected components in $\mathcal{I}\setminus \mathcal{U}$. However, with the conditions of Theorem \ref{theorem2} the set $\mathcal{I}\setminus \mathcal{U}$ is connected (see the proof) and a unique point in  $\mathcal{I}\setminus \mathcal{U}$ is required $(K = 1)$.\\
Proposition \ref{infeasible} remains the same with $\mathbb{B}$ instead of $\mb{B}$. The detection of the inclusion $\mathcal{I} \subset \mathbb{B}$ is now made possible and is important for sequential problems, since in this case, the ball $\mathbb{B}$ is not added to the set $\Lambda$.

\section*{Acknowledgments}
I would like to thank my colleagues Guillem Rigaill (Evry University), Michel Koskas (AgroParisTech) and Julien Chiquet (AgroParisTech) for relevant comments and for their encouraging attitude regarding this work.

\bibliographystyle{siamplain}
\bibliography{references}
\end{document}